\documentclass[pra,amsmath,amsfonts,twocolumn,superscriptaddress,showpacs,showkeys,article]{revtex4-1}
 \usepackage{epsf,epsfig}
 \usepackage[psamsfonts]{amssymb}
\usepackage{amsthm}
\usepackage{graphicx}
\usepackage{color,soul}
\usepackage{xcolor}

\newtheorem{theorem}{Theorem}
\newtheorem{lemma}[theorem]{Lemma}

\newtheorem{corollary}[theorem]{Corollary}

 \newcommand{\ket}[1]{|#1\rangle}
 \newcommand{\bra}[1]{\langle #1|}

 \newcommand{\Dim}{\mathrm{Dim}}

\newcommand{\Tr}{{\mathrm {Tr}}}

\newcommand{\etal}{\textit{et al.} }

\begin{document}
\preprint{APS/123-QED}

\title{Reexamination  of strong subadditivity: A quantum-correlation approach
}

\author{Razieh Taghiabadi}
\author{Seyed Javad Akhtarshenas}
\email{akhtarshenas@um.ac.ir}
\author{Mohsen Sarbishaei}
\affiliation{Department of Physics, Ferdowsi University of Mashhad, Mashhad, Iran}

\begin{abstract}
The  strong subadditivity inequality  of von Neumann entropy  relates the  entropy of subsystems of a tripartite state  $\rho_{ABC}$ to that of the composite system. Here, we define $\boldsymbol{T}^{(a)}(\rho_{ABC})$ as the extent to which $\rho_{ABC}$ fails to satisfy  the strong subadditivity inequality $S(\rho_{B})+S(\rho_{C}) \le S(\rho_{AB})+S(\rho_{AC})$ with equality and investigate its properties. In particular, by introducing auxiliary subsystem $E$, we consider any  purification $|\psi_{ABCE}\rangle$ of $\rho_{ABC}$ and formulate  $\boldsymbol{T}^{(a)}(\rho_{ABC})$ as  the extent to which the  bipartite quantum correlations of $\rho_{AB}$ and $\rho_{AC}$, measured by entanglement of formation and quantum discord, change under the transformation $B\rightarrow BE$ and $C\rightarrow CE$.  Invariance  of quantum  correlations of $\rho_{AB}$ and $\rho_{AC}$  under such transformation is shown to be a necessary and sufficient condition for vanishing $\boldsymbol{T}^{(a)}(\rho_{ABC})$.  Our approach allows one to characterize, intuitively,  the structure of states for which the strong subadditivity is saturated. Moreover,  along with providing a  conservation law for quantum correlations of states  for which the strong subadditivity inequality is satisfied with equality, we find that such states  coincides with those that   the Koashi-Winter monogamy relation is saturated.
\end{abstract}
%
%
\maketitle

\textit{Introduction.---}Correlations between different subsystems of a classical or quantum composite system result to inequalities relating the  entropy  of various subsystems to that of the composite  system.     For a given state $\rho_{AB}$ of the quantum system $\mathcal{H}_{AB}$, consisting of two subsystems  $\mathcal{H}_{A}$  and $\mathcal{H}_{B}$, the \emph{subadditivity} (SA) states that \cite{WehrlRMP1978}
\begin{equation}\label{SA}
S(\rho_{AB})\le S(\rho_{A})+S(\rho_{B}),
\end{equation}
where $\rho_{A}=\Tr_{B}{(\rho_{AB})}$ and  $\rho_{B}=\Tr_{A}{(\rho_{AB})}$ are states of the subsystems and $S(\rho)=-\Tr{\rho\log{\rho}}$  is the von Neumann entropy of $\rho$. The  inequality implies that any correlation between subsystems decreases the amount of information needed to specify one of the subsystems once we know the other one \cite{ZyczkowskiBook2006}. The equality holds if and only if the subsystems are uncorrelated, i.e., $\rho_{AB}=\rho_{A}\otimes \rho_{B}$. It turns out therefore that the extent to which the state $\rho_{AB}$ fails to satisfy the SA with  equality is a measure of total correlations (classical+quantum) and is  defined as the \emph{mutual information} \cite{NielsenBook2000}
\begin{equation}\label{MI}
I(\rho_{AB})=S(\rho_{A})+S(\rho_{B})-S(\rho_{AB}).
\end{equation}
A stronger  inequality holds  when the composite system is composed of three subsystems. Suppose $\rho_{ABC}$ is a  quantum state of the composite system  $\mathcal{H}_{ABC}=\mathcal{H}_{A}\otimes\mathcal{H}_{B}\otimes\mathcal{H}_{C}$. The \emph{strong subadditivity} (SSA),  which was conjectured for quantum systems by Lanford and Robinson \cite{LanfordJMP1968} and  proved by Lieb and Ruskai \cite{LiebRuskaiJMP1975}, states that
\begin{equation}\label{SSA1}
S(\rho_{ABC})+S(\rho_{C}) \le S(\rho_{AC})+S(\rho_{BC}).
\end{equation}
Clearly, by choosing $\mathcal{H}_C=\mathbb{C}$,  Eq. \eqref{SSA1}  recovers the  subadditivity relation \eqref{SA}.
The SSA inequality \eqref{SSA1} is equivalent to
\begin{equation}\label{SSA2}
S(\rho_{B})+S(\rho_{C}) \le S(\rho_{AB})+S(\rho_{AC}),
\end{equation}
which can  be seen  by adding an auxiliary  subsystem $\mathcal{H}_{E}$ such that $\rho_{ABCE}$ is a pure state.
The SSA inequality plays a crucial role in quantum information theory. It appears almost everywhere in quantum information theory from Holevo bound on the accessible information in the quantum ensemble \cite{HolevoPIT1973,SchumacherPRL1996,RogaPRL2010}, properties of coherent information \cite{SchumacherPRA1996,BarnumPRA1998,LloydPRA1997}, definition of squashed entanglement \cite{ChristandlJMP2004,CarlenLMP2012},  monogamy of quantum correlations \cite{KoashiPRA2004,FanchiniPRA2011,YangPRA2013}, to   some features of quantum discord such as  condition for nullity \cite{DattaARXIV1003.5256}, and condition for saturating the upper bound \cite{XiPRA2012}.  Note that the SSA inequality \eqref{SSA2} can be written also as \cite{NielsenBook2000}
\begin{equation}\label{SSA2p}
0 \le S_{\rho_{AB}}(A \mid B)+S_{\rho_{AC}}(A \mid C),
\end{equation}
where
\begin{equation}\label{CE}
S_{\rho_{AB}}(A \mid B)=S(\rho_{AB})-S(\rho_{B}),
\end{equation}
is the \emph{quantum conditional entropy}  which, unlike its classical counterpart, can take negative values.   This latter form of SSA  emphasises how highly SSA is nontrivial in the quantum case in a sense that although each term on the right-hand side of \eqref{SSA2p} can be negative, both of them cannot be negative simultaneously.

Contrary to the SA inequality, the characterization of states for which the SSA is saturated is not  trivial. Obviously, when the global state is factorized, i.e.,  $\rho_{ABC}=\rho_{A}\otimes \rho_{B}\otimes \rho_{C}$, the equality \eqref{SSA1} holds but the converse is not true in general. Petz  \cite{PetzCMP1986} and Ruskai \cite{RuskaiJMP2002} have provided algebraic criteria to check that if any given state satisfies the inequality with equality, however, their description  do not characterize the structure of such states. An important progress in providing the structure of states for which the SSA inequality is satisfied with equality is given in  Ref.  \cite{HaydenCMP2004}. Hayden \etal have shown that the state $\rho_{ABC}$ satisfies the SSA inequality \eqref{SSA1} with equality if and only if there is a decomposition of the subsystem $\mathcal{H}_C$ into a direct orthogonal sum of tensor products as $\mathcal{H}_C=\bigoplus_{j} \mathcal{H}_{C_j^L}\otimes \mathcal{H}_{C_j^R}$ such that $\rho_{ABC}=\bigoplus_{j} q_j\rho_{AC_j^L}\otimes \rho_{C_j^RB}$, with states $\rho_{AC_j^L}$ on $\mathcal{H}_A\otimes \mathcal{H}_{C_j^L}$ and $\rho_{C_j^RB}$ on  $\mathcal{H}_{C_j^R}\otimes \mathcal{H}_B$, and a probability distribution $\{q_j\}$. On the basis of the results of \cite{HaydenCMP2004}, the structure of states for which the SSA  inequality \eqref{SSA2} is saturated  is given in \cite{ZhangCTP2015}.

As it is mentioned above any deviation of SA from equality  refers to some correlations existing in the bipartite state. It is therefore natural to ask the question: Is it possible to express deviation from equality of the SSA inequality  \eqref{SSA2} to the existence of some kind of correlations.    In this paper, we address this issue and define the quantity $\boldsymbol{T}^{(a)}(\rho_{ABC})$ as  the extent to which the SSA inequality \eqref{SSA2} deviates from equality   in terms of   two different aspects of quantum correlations, i.e., entanglement of formation (EOF) \cite{BennettPRA1996} and quantum discord (QD)  \cite{ZurekPRL2001,HendersonJPA2001}.  Our main results are Theorems \ref{TheoremDeltaCorrelations} and \ref{TheoremSSAcondition} and also Corollary \ref{Corollary-Results}. To be specific,  Theorem \ref{TheoremDeltaCorrelations}  shows that  $\boldsymbol{T}^{(a)}(\rho_{ABC})$ can be expressed  by means of  the extent to which the  bipartite quantum correlations of $\rho_{AB}$ and $\rho_{AC}$, measured by EOF and QD, change under the transformation $B\rightarrow BE$ and $C\rightarrow CE$,   where $E$ is  an auxiliary subsystem  that purifies $\rho_{ABC}$. In Theorem \ref{TheoremSSAcondition}, we use this measure and characterize the structure of states for which the SSA inequality \eqref{SSA2} is satisfied with equality. Our approach provides an information-theoretic aspect for such states, that is, $\boldsymbol{T}^{(a)}(\rho_{ABC})=0$ if and only if quantum  correlations of $\rho_{AB}$ and $\rho_{AC}$  do not change under the above transformation.  This, however, can be regarded as a kind of  conservation law for quantum correlations, i.e.,  if $\boldsymbol{T}^{(a)}(\rho_{ABC})=0$ then the sum of entanglement of formation of $\rho_{AB}$ and $\rho_{AC}$ is equal to the sum of  their quantum discord. Moreover, we find that the class of states saturating the Koashi-Winter inequality  \eqref{KWinequality} coincides with those states that the SSA inequality is satisfied with equality.

\textit{Monogamy of quantum correlations.---}Quantum entanglement and quantum discord are two aspects of quantum correlations defined respectively within entanglement-separability paradigm and an information-theoretic perspective. Quantum entanglement is defined as those correlations that
cannot be generated by local operations and classical communication \cite{WernerPRA1989}. Several measures have been proposed to quantify quantum entanglement,  the most important one  is the \emph{entanglement of formation} (EOF) \cite{WoottersPRL1998}, defined for a bipartite state $\rho_{AB}$ as
\begin{equation}
E(\rho_{AB})=\min\sum_{i}p_i E(\psi_i),
\end{equation}
where minimum is taken over all pure state decompositions $\rho_{AB}=\sum_{i}p_i\ket{\psi_i}\bra{\psi_i}$ and $E(\psi_i)=S(\Tr_{B}[\ket{\psi_i}\bra{\psi_i}])$.
However, quantum entanglement cannot capture all nonclassical correlations of a composite state in a sense that  a composite mixed  state may exhibit some quantum correlations even if it is disentangled. From various measures proposed for this different aspect of quantum correlation, \emph{quantum discord} (QD)  has received a great deal of attention. For a bipartite state $\rho_{AB}$,   quantum discord is defined as the difference between two classically equivalent  but quantum mechanically different definitions of quantum mutual information
\begin{equation}\label{QD}
D^{(B)}(\rho_{AB})=I(\rho_{AB})-J^{(B)}(\rho_{AB}),
\end{equation}
where $J^{(B)}(\rho_{AB})=\max_{\{\Pi_i^B\}}J^{\{\Pi_i^{B}\}}(\rho_{AB})$  is the classical correlation of the state $\rho_{AB}$. Moreover  \cite{ZurekPRL2001}
\begin{equation}\label{CC}
J^{\{\Pi_i^{B}\}}(\rho_{AB})=S(\rho_{A})-\sum_ip_iS(\rho_{AB}|\Pi_i^{B}),
\end{equation}
where  $\{\Pi_i^{B}\}$ is the set of projection operators on the subsystem $B$, and  $S(\rho_{AB}|\Pi_i^{B})$ is the conditional entropy of $A$ when measurement is performed on $B$ and the $i$-th outcome is obtained with probability $p_i=\mathrm{Tr}[\Pi_i^{B}\rho_{AB}\Pi_i^{B}]$.  We notice here that  QD is not,  in general,  symmetric under the swap of the two parties, $A\leftrightarrow B$.

For a general tripartite state $\rho_{ABC}$ there exists a monogamic  relation   between EOF and QD  of its corresponding bipartite mixed  states, the so-called Koashi-Winter (K-W) relation \cite{KoashiPRA2004}
\begin{equation}\label{KWinequality}
E(\rho_{AB})\leq D^{(C)}(\rho_{AC})+S_{\rho_{AC}}({A\mid C}).
\end{equation}
When $\rho_{ABC}$ is  pure, the inequality is saturated
and   $\rho_{ABC}$ is called a purification  of the mixed states $\rho_{AB}$ and $\rho_{AC}$. In this case  the state $\rho_{AC}$ is called $B$-complement to $\rho_{AB}$ and, similarly, $\rho_{AB}$ is called $C$-complement to $\rho_{AC}$. For this particular case of pure global state $\rho_{ABC}$, the following \emph{quantum conservation law}, as it is called by Fanchini \etal   \cite{FanchiniPRA2011}, is obtained
\begin{equation}\label{ConservationLaw}
E(\rho_{AB})+E(\rho_{AC})=D^{(B)}(\rho_{AB})+D^{(C)}(\rho_{AC}).
\end{equation}
Moreover, Cen \etal \cite{CenPRA2011} have used the K-W relation and proposed a schema to  quantify the  QD and EOF and their ordering relation. In particular, they have characterized the QD of an arbitrary two-qubit state reduced from pure three-qubit states and a class of rank 2 mixed states of $4\times 2$ systems.  In \cite{BaiPRA2013}, the authors used the  K-W relation and explored the monogamy property of the square of QD in mutipartite systems. They have shown that the square QD is monogamous for three-qubit pure states.

\textit{Deviation from equality  of strong subadditivity.---}Let $\boldsymbol{T}^{(a)}(\rho_{ABC})$ denotes  the degree to which $\rho_{ABC}$ fails to saturate  SSA, i.e., the difference between the right-hand side and the left-hand side of Eq. \eqref{SSA2}
\begin{eqnarray}\label{DeltaARight-Left}
\boldsymbol{T}^{(a)}(\rho_{ABC})=S(\rho_{AB})+S(\rho_{AC})-S(\rho_{B})-S(\rho_{C}).
\end{eqnarray}
Then the following Lemma states that $\boldsymbol{T}^{(a)}(\rho_{ABC})$ is a concave function of its input state $\rho_{ABC}$.
\begin{lemma}\label{LemmaDeltaConcavity}
Suppose  $\{p_k,\rho_{ABC}^k\}_{k}$ is an ensemble of states generating $\rho_{ABC}$, that is, $\rho_{ABC}=\sum_{k}^Kp_k\rho_{ABC}^k$. Then
\begin{equation}\label{DeltaConcavity}
\boldsymbol{T}^{(a)}(\rho_{ABC})\ge \sum_{k}^Kp_k \boldsymbol{T}^{(a)}(\rho_{ABC}^k).
\end{equation}
Moreover, the equality holds  if the marginal  states  $\{\rho_{B}^k\}_k$ and $\{\rho_{C}^k\}_k$ are mutually  orthogonal, i.e.,  $\rho_{B}^k\perp \rho_{B}^{k^\prime}$ and $\rho_{C}^k\perp \rho_{C}^{k^\prime}$ for $k\ne k^\prime$.
\end{lemma}
\begin{proof}
For state $\rho=\sum_{k}^Kp_k\rho^k$, acting on $\mathcal{H}$, define the \emph{Holevo quantity} $\chi_{\rho}=S(\rho)-\sum_kp_kS(\rho^k)$. The inequality \eqref{DeltaConcavity} then easily obtained by noting that for any bipartite state $\rho_{AB}$ we have $\chi_{\rho_{AB}}\ge \chi_{\rho_B}$ \cite{SchumacherPRL1996,NielsenBook2000} (For a  different proof of \eqref{DeltaConcavity} see Ref.  \cite{NielsenBook2000}).    For the second claim of the Lemma note that  for any state $\rho=\sum_{k}p_k\rho^k$ we have  $S\left(\sum_{k}p_k\rho^k\right)\le H(p_k)+\sum_{k}p_kS(\rho^k)$, with equality  if and only if  the states $\{\rho^k\}_k$ have support in orthogonal subspaces $\mathcal{H}_{k}$ of the  Hilbert space $\mathcal{H}=\bigoplus_{k=1}^{K}\mathcal{H}_{k}$ \cite{NielsenBook2000}, i.e., they are mutually orthogonal in a sense that $\rho^k\perp \rho^{k^\prime}$ for $k\ne k^\prime$. Using this we find that the orthogonality condition for the marginal states  $\{\rho_{AB}^k\}_k$, $\{\rho_{AC}^k\}_k$, $\{\rho_{B}^k\}_k$, and $\{\rho_{C}^k\}_k$ implies the equality of the inequality \eqref{DeltaConcavity}. The proof becomes complete if recalling that for any two bipartite states $\rho_{AB}$ and $\rho_{AB}^\prime$, the orthogonality of parts implies the orthogonality of whole, i.e.,  $\rho_{B}\perp\rho_{B}^\prime$ implies $\rho_{AB}\perp\rho_{AB}^\prime$.
\end{proof}
The following corollary is immediately obtained from Lemma  \ref{LemmaDeltaConcavity}.
\begin{corollary}\label{CorollaryZeroDelta}
(i) Let $\rho_{ABC}$ be a state with vanishing $\boldsymbol{T}^{(a)}(\rho_{ABC})$, i.e., $\boldsymbol{T}^{(a)}(\rho_{ABC})=0$. Then  for any ensemble  $\{p_k,\rho_{ABC}^k\}_{k}$ giving rise to $\rho_{ABC}$ we have  $\boldsymbol{T}^{(a)}(\rho_{ABC}^k)=0$ for $k=1,\cdots,K$. (ii) Moreover,  any ensemble $\{p_k,\rho_{ABC}^k\}_{k}$ with $\boldsymbol{T}^{(a)}(\rho_{ABC}^k)=0$ generates a state $\rho_{ABC}=\sum_{k}^{K}p_k\rho_{ABC}^k$ with $\boldsymbol{T}^{(a)}(\rho_{ABC})=0$ provided that the  marginal  states  $\{\rho_{B}^k\}_k$ and $\{\rho_{C}^k\}_k$ are mutually  orthogonal.
\end{corollary}
\textit{Strong subadditivity versus quantum correlations.---}Consider two bipartite states $\rho_{AB}$ and $\rho_{AC}$ with the same tripartite extension $\rho_{ABC}$, i.e., $\rho_{AB}=\Tr_{C}(\rho_{ABC})$ and $\rho_{AC}=\Tr_{B}(\rho_{ABC})$. Let $\ket{\psi_{ABCE}}$ be any purification of $\rho_{ABC}$, where $E$ is an auxiliary subsystem. Define $\widetilde{B}=BE$ and $\widetilde{C}=CE$, i.e., $\mathcal{H}_{\widetilde{B}}=\mathcal{H}_B\otimes \mathcal{H}_E$ and $\mathcal{H}_{\widetilde{C}}=\mathcal{H}_C\otimes \mathcal{H}_E$.  Armed with these definitions,  let us apply   the K-W relation \eqref{KWinequality} to the  pure states $\rho_{AB\widetilde{C}}$ and $\rho_{A\widetilde{B}C}$ and  the mixed state $\rho_{ABC}$ to  get
\begin{eqnarray}\label{KWrhoC}
E(\rho_{AB})&=&D^{(\widetilde{C})}(\rho_{A\widetilde{C}})+S_{\rho_{A\widetilde{C}}}(A \mid \widetilde{C}),
\\ \label{KWrhoB}
E(\rho_{A\widetilde{B}})&=&D^{(C)}(\rho_{AC})+S_{\rho_{AC}}(A \mid C),
\\ \label{KWrho}
E(\rho_{AB})&\le &D^{(C)}(\rho_{AC})+S_{\rho_{AC}}(A \mid C).
\end{eqnarray}
The same relations hold if we exchange $B \leftrightarrow C$; denoting them with (\ref{KWrhoC}$'$), (\ref{KWrhoB}$'$), and (\ref{KWrho}$'$), respectively.
Theorem \ref{TheoremDeltaCorrelations} provides a relation for $\boldsymbol{T}^{(a)}(\rho_{ABC})$ in terms of quantum correlations.
\begin{theorem}\label{TheoremDeltaCorrelations}
$\boldsymbol{T}^{(a)}(\rho_{ABC})$ can be expressed in terms of the quantum correlations of the  aforementioned bipartite states
\begin{eqnarray}\label{DeltaACorrelations}\nonumber
&&\hspace{-8mm}\boldsymbol{T}^{(a)}(\rho_{ABC}) \\ \nonumber
&=&\left[E(\rho_{A\widetilde{B}})-E(\rho_{AB})\right]+\left[D^{(\widetilde{C})}(\rho_{A\widetilde{C}})- D^{(C)}(\rho_{AC})\right] \\  \nonumber
&=&\left[E(\rho_{A\widetilde{C}})-E(\rho_{AC})\right]+\left[D^{(\widetilde{B})}(\rho_{A\widetilde{B}})- D^{(B)}(\rho_{AB})\right] \\  \nonumber
&=&\left[E(\rho_{A\widetilde{B}})+E(\rho_{A\widetilde{C}})\right]-\left[D(\rho_{AB})+D(\rho_{AC})\right] \\ \label{DeltaACorrelations}
&=& \left[D(\rho_{A\widetilde{B}})+D(\rho_{A\widetilde{C}})\right]-\left[E(\rho_{AB})+E(\rho_{AC})\right].
\end{eqnarray}
Moreover
\begin{eqnarray}\label{deltaE}
\delta^{(a)}_x(E)=E(\rho_{A\widetilde{X}})-E(\rho_{AX}) &\ge & 0,
\end{eqnarray}
for $X=B,C$.
\end{theorem}
\begin{proof}
By subtracting Eq. \eqref{KWrhoC} from  \eqref{KWrhoB}, and using that for  any pure state $\rho_{XYZ}$ one can write $S_{\rho_{XZ}}(X \mid Z)=-S_{\rho_{XY}}(X \mid Y)$,   we find the first equality of Eq. \eqref{DeltaACorrelations}. The third equality is obtained by adding Eq. \eqref{KWrhoB} to (\ref{KWrhoB}$'$). The second and fourth equalities  are obtained just by  swapping $B\leftrightarrow C$.  Moreover, the nonnegativity of Eq. \eqref{deltaE}  follows easily by subtracting Eq.  \eqref{KWrho} from \eqref{KWrhoB}.
\end{proof}
Theorem  \eqref{TheoremDeltaCorrelations}  claims that for a given $\rho_{ABC}$ with \emph{a prior} bipartite quantum correlations  of $A$ with other subsystems $B$ and $C$,    $\boldsymbol{T}^{(a)}(\rho_{ABC})$ quantifies the increased quantum correlations  caused by the transformations $B\rightarrow\widetilde{B}=BE$ and $C\rightarrow\widetilde{C}=CE$.   This feature of $\boldsymbol{T}^{(a)}(\rho_{ABC})$ allows one to  find, along with other results, the structure of states for which $\boldsymbol{T}^{(a)}(\rho_{ABC})$  achieves  its  upper or lower bounds. For the former, note that
\begin{eqnarray}\label{DeltaUpperBound}
\boldsymbol{T}^{(a)}(\rho_{ABC}) & \le &
\max\left[E(\rho_{A\widetilde{B}})+E(\rho_{A\widetilde{C}})\right] \\ \nonumber
&-&\min\left[D(\rho_{AB})+D(\rho_{AC})\right] \\ \nonumber
&\le & 2S(\rho_{A})\le 2\log{d_A},
\end{eqnarray}
which obtained from  the third equality of  Eq. \eqref{DeltaACorrelations} and the fact that for any bipartite state $\rho_{XY}$, we have $E(\rho_{XY})\le \min\{S(\rho_{X}),S(\rho_{Y})\}$ and $D^{(Y)}(\rho_{XY})\ge 0$. A simple investigation shows that the last inequality is saturated if and only if
\begin{equation}\label{DeltaUpperBoundRho}
\rho_{ABC}=\mathbb{I}_{d_A}/d_{A}\otimes \rho_{BC},
\end{equation}
where $\mathbb{I}_{d_A}$ denotes the unity  matrix of $\mathcal{H}_A$ and $d_{A}=\dim{\mathcal{H}_A}$.  This  follows from the fact that the minimum in Eq. \eqref{DeltaUpperBound} is obtained  when the prior quantum discord of $A$ with $B$ and $C$ are zero, and the maximum is achieved when the transformed states $\rho_{A\widetilde{B}}$ and $\rho_{A\widetilde{C}}$ are maximally entangled states. Both extermums will be attained  simultaneously  if and only if $\rho_{A}$ is maximally mixed state  and factorized from the rest of the system,  so that  $A$  does not possess  any prior correlation with $B$ and $C$ at all.       Although Eq. \eqref{DeltaUpperBoundRho}   can be obtained   by means of  Eqs.  \eqref{SA} and   \eqref{SSA2p}, the above investigation shows the role of quantum correlations more clearly.

\textit{Conditions for $\boldsymbol{T}^{(a)}(\rho_{ABC})=0$.---}Theorem \ref{TheoremDeltaCorrelations} states that  a state satisfies SSA with equality if and only if the quantum correlations of $A$ with $B$ and $C$ do not change under the transformation $B\rightarrow\widetilde{B}$ and $C\rightarrow\widetilde{C}$. This happens, in particular, whenever $\rho_{ABC}$ is pure, but this is not the only case that $\boldsymbol{T}^{(a)}(\rho_{ABC})$ vanishes. Indeed, it is not difficult to see that any state of the form $\rho_{ABC}=\ket{\psi_{AY}}\bra{\psi_{AY}}\otimes \rho_{Z}$     possesses quantum correlations which are invariant under transformations $B\rightarrow \widetilde{B}$ and $C\rightarrow \widetilde{C}$.
\begin{lemma}\label{LemmaSSAcondition}
Let  $Y$ be a partition of $BC$ and $Z$ its complement, i.e., $\mathcal{H}_{Y}\otimes\mathcal{H}_{Z}=\mathcal{H}_{B}\otimes \mathcal{H}_{C}$. More precisely, let $\mathcal{H}_{B}$ and $\mathcal{H}_{C}$ can be tensor producted as $\mathcal{H}_{B}=\mathcal{H}_{B^L}\otimes\mathcal{H}_{B^R}$ and $\mathcal{H}_{C}=\mathcal{H}_{C^L}\otimes\mathcal{H}_{C^R}$, respectively, and define $\mathcal{H}_{Y}=\mathcal{H}_{B^L}\otimes\mathcal{H}_{C^L}$ and $\mathcal{H}_{Z}=\mathcal{H}_{B^R}\otimes\mathcal{H}_{C^R}$.  Then for any state of the form
\begin{equation}\label{PsiRhoProduct1}
\rho_{ABC}=\ket{\psi_{AY}}\bra{\psi_{AY}}\otimes \rho_{Z},
\end{equation}
we have $\boldsymbol{T}^{(a)}(\rho_{ABC})=0$.
\end{lemma}
\begin{proof}
Recall that purification  produces entanglement between the system under consideration and the auxiliary system  if and only if  the original system is impure \cite{NielsenBook2000}.  Any purification of $\rho_{ABC}$ leads to a purification of $\rho_{Z}$ as $\rho_{\widetilde{Z}}$ where $\widetilde{Z}=ZE$ for the  auxiliary subsystem $E$.
In turn, it change $\rho_{ABC}$ to $\rho_{ABCE}=\ket{\psi_{AY}}\bra{\psi_{AY}}\otimes \rho_{\widetilde{Z}}$.  Clearly, this does not create any correlation between the first and the second part of $\rho_{ABC}$.
\end{proof}
Although the above Lemma provides a sufficient condition for equality of SSA, it is not necessary in general.  However, it provides building blocks for the structure of states  for which the  SSA inequality is satisfied with equality. Indeed, for a fixed pure state $\ket{\psi_{AY}}$, it is not difficult to see that the set of states with structure given by Eq. \eqref{PsiRhoProduct1} forms a convex subset of the set of all states. It follows therefore  that \emph{only} ensembles of the form $\{p_k,\ket{\psi^k_{AY_k}}\bra{\psi^k_{AY_k}}\otimes \rho^k_{Z_k}\}_{k}$ can realize a state with $\boldsymbol{T}^{(a)}(\rho_{ABC})=0$. Theorem \ref{TheoremSSAcondition} provides the  necessary and sufficient conditions on $Y_k$ and $Z_k$, in order to achieve states with $\boldsymbol{T}^{(a)}(\rho_{ABC})=0$.
\begin{theorem}\label{TheoremSSAcondition}
For a given $\rho_{ABC}$ we have $\boldsymbol{T}^{(a)}(\rho_{ABC})=0$ if and only if  $\rho_{ABC}$ can be expressed as
\begin{equation}\label{RhoABC-RhoABCk}
\rho_{ABC}=\sum_{k=1}^Kp_k\rho_{ABC}^k,
\end{equation}
such that the marginal states  $\{\rho_{B}^k\}_k$ and $\{\rho_{C}^k\}_k$ are mutually orthogonal and  $\rho_{ABC}^k$ has the following form
\begin{equation}\label{PsiRhoProductK}
\rho_{ABC}^k=\ket{\psi^k_{AY_k}}\bra{\psi^k_{AY_k}}\otimes \rho^k_{Z_k}.
\end{equation}
Here $\rho_{ABC}^k$ is defined on $\mathcal{H}_{A}\otimes \mathcal{H}_{B_k}\otimes \mathcal{H}_{C_k}$, and  $Y_k$ is a partition of $B_kC_k$ and $Z_k$ denotes its complement in such a way that $\mathcal{H}_{Y_k}=\mathcal{H}_{B_k^L}\otimes \mathcal{H}_{C_k^L}$ and $\mathcal{H}_{Z_k}=\mathcal{H}_{B_k^R}\otimes \mathcal{H}_{C_k^R}$.
\end{theorem}
\begin{proof}
The sufficient condition is a simple consequence of Corollary \ref{CorollaryZeroDelta}-(ii) and Lemma \ref{LemmaSSAcondition}. To prove the necessary condition,  let $\rho_{ABC}$ be a state such that $\boldsymbol{T}^{(a)}(\rho_{ABC})=0$. It follows from Corollary \ref{CorollaryZeroDelta}-(i) that  $\rho_{ABC}=\sum_k p_k\rho_{ABC}^k$ implies $\boldsymbol{T}^{(a)}(\rho_{ABC}^k)=0$ for $k=1,\cdots,K$, as such,  $\rho_{ABC}^k$  takes the form given by \eqref{PsiRhoProductK}. It remains only to prove   that $\rho_{B}^k \perp \rho_{B}^{k^\prime}$ and  $\rho_{C}^k \perp \rho_{C}^{k^\prime}$ for $k\ne k^\prime$. Using the spectral decomposition $\rho_{Z_k}^k=\sum_{i_k}\mu_{k}^{i_k}\ket{\mu_{Z_k}^{i_k}}\bra{\mu_{Z_k}^{i_k}}$ and introducing the auxiliary subsystem $E$ with an orthogonal decomposition $\mathcal{H}_E=\bigoplus_{k}^{K}\mathcal{H}_{E_k}$ and corresponding orthonormal basis $\{\ket{\lambda_{E_k}^{i_k}}\}_{i_k}$, for $k=1,\cdots,K$, we arrive at the  following  purification for $\rho_{ABC}$
\begin{equation}\label{PsiABCE}
\ket{\Psi_{ABCE}}=\sum_{k}^{K}\sqrt{p_k}\ket{\psi^k_{AY_k}}\otimes \ket{\phi^k_{Z_kE_k}},
\end{equation}
where $\ket{\phi^k_{Z_kE_k}}=\sum_{i_k}\sqrt{\mu_{k}^{i_k}}\ket{\mu_{Z_k}^{i_k}}\ket{\lambda_{E_k}^{i_k}}$. Tracing out $C$ and defining $\widetilde{B}=BE$, we find
\begin{eqnarray}\label{RhoABE} \nonumber
\rho_{A\widetilde{B}}&=&\sum_{k,k^\prime}^{K}\sqrt{p_kp_{k^\prime}}
\Tr_{C}\left[\ket{\psi^k_{AY_k}}\bra{\psi^{k^\prime}_{AY_{k^\prime}}}\otimes \ket{\phi^k_{Z_kE_k}}\bra{\phi^{k^\prime}_{Z_{k^\prime}E_{k^\prime}}}\right] \\
 \nonumber
&=&\sum_{k\ne k^\prime}^{K}\sqrt{p_kp_{k^\prime}}
\Tr_{C}\left[\ket{\psi^k_{AY_k}}\bra{\psi^{k^\prime}_{AY_{k^\prime}}}\otimes \ket{\phi^k_{Z_kE_k}}\bra{\phi^{k^\prime}_{Z_{k^\prime}E_{k^\prime}}}\right]
\\ \label{RhoABE}
&+&\sum_{k}^{K}p_k \sigma_{AB_k^L}^k\otimes \varrho_{B_k^RE_k}^k.
\end{eqnarray}
Here  we have defined $\sigma_{AB_k^L}^k=\Tr_{C_k^L}\ket{\psi^k_{AY_k}}\bra{\psi^k_{AY_k}}$ and  $\varrho_{B_k^RE_k}^k=\Tr_{C_k^R}\left[\ket{\phi^k_{Z_kE_k}}\bra{\phi^{k}_{Z_{k}E_{k}}}\right]$ as states on $\mathcal{H}_{A}\otimes \mathcal{H}_{B_k^L}$ and $\mathcal{H}_{B_k^R}\otimes \mathcal{H}_{E}$, respectively. Now, tracing out $C$ from $\rho_{ABC}$ of Eq. \eqref{RhoABC-RhoABCk}, we get
\begin{eqnarray}\label{RhoAB}
\rho_{AB}=\sum_{k}^{K}p_k
\sigma_{AB_k^L}^k \otimes \varrho_{B_k^R}^k,
\end{eqnarray}
where  $\varrho_{B_k^R}^k=\Tr_{C_{k}^R}\rho_{Z_k}^k$. Comparing this with $\rho_{A\widetilde{B}}$ given by Eq. \eqref{RhoABE},
one can easily see that $E(\rho_{A\widetilde{B}})=E(\rho_{AB})$ if and only if $\rho_{A\widetilde{B}}$ can be written as
\begin{eqnarray}\label{RhoABE2}
\rho_{A\widetilde{B}}=\sum_{k}^{K}p_k
\sigma_{AB_k^L}^k \otimes \varrho_{B_k^RE_k}^k.
\end{eqnarray}
This happens if and only if the first term of the second equality of  \eqref{RhoABE} vanishes, i.e., $\Tr_{C}\left[\ket{\psi^k_{AY_k}}\bra{\psi^{k^\prime}_{AY_{k^\prime}}}\otimes \ket{\phi^k_{Z_kE_k}}\bra{\phi^{k^\prime}_{Z_{k^\prime}E_{k^\prime}}}\right]=0$. This, in turn implies     $\rho_{C}^k \perp \rho_{C}^{k^\prime}$ for $k\ne k^\prime$,
where $\rho_{C}^k=\Tr_{AB}[\ket{\psi^k_{AY_k}}\bra{\psi^k_{AY_k}}\otimes\ket{\mu_{Z_k}^{i_k}}\bra{\mu_{Z_k}^{i_k}}]$ and $\rho_{C}^{k^\prime}$ is defined similarly.
In a same manner we find that the condition $E(\rho_{A\widetilde{C}})=E(\rho_{AC})$ leads to $\rho_{B}^k \perp \rho_{B}^{k^\prime}$ for $k\ne k^\prime$. Using this,  it follows that  $D^{(\widetilde{B})}(\rho_{A\widetilde{B}})=D^{(B)}(\rho_{AB})$, which completes the proof.
\end{proof}
The following Corollary is immediately obtained from Theorem \ref{TheoremSSAcondition}.
\begin{corollary}\label{Corollary-Results}
(i) $\rho_{ABC}$ satisfies the SSA inequality \eqref{SSA2} with equality if and only if  it satisfies the Koashi-Winter relation \eqref{KWinequality} with equality, i.e., $S(\rho_{B})+S(\rho_{C}) = S(\rho_{AB})+S(\rho_{AC})$ implies $E(\rho_{AB})= D^{(C)}(\rho_{AC})+S_{\rho}({A\mid C})$ and $E(\rho_{AC})= D^{(B)}(\rho_{AB})+S_{\rho}({A\mid B})$, and vice versa. \\
(ii)  If $\rho_{ABC}$ satisfies the SSA inequality \eqref{SSA2} with equality then it satisfies  the quantum conservation law, i.e., $S(\rho_{B})+S(\rho_{C}) = S(\rho_{AB})+S(\rho_{AC})$ implies  $E(\rho_{AB})+E(\rho_{AC})=D^{(B)}(\rho_{AB})+D^{(C)}(\rho_{AC})$.
\end{corollary}
Note that we have considered only situation  that the SSA inequality is saturated identically. Indeed, as it is clear from the first line of Eq. \eqref{DeltaACorrelations},  vanishing $\boldsymbol{T}^{(a)}(\rho_{ABC})$ may happen even for nonzero $E(\rho_{A\widetilde{B}})- E(\rho_{AB})$, due to the possibility that   $D^{(\widetilde{C})}(\rho_{A\widetilde{C}})- D^{(C)}(\rho_{AC})$ takes negative value. This \emph{approximate} case that a state \emph{almost} saturates  SSA inequality   is also addressed   in Ref. \cite{HaydenCMP2004} for the SSA inequality \eqref{SSA1}, and it is proved in \cite{BrandaoCMP2011} that in this case $\rho_{ABC}$ is well approximated by structure given in \cite{HaydenCMP2004}.  However, a look at Eq.  \eqref{DeltaACorrelations} shows that  such approximate case does not happen if  $\delta^{(a)}_x(D)=D^{(\widetilde{X})}(\rho_{A\widetilde{X}})- D^{(X)}(\rho_{AX})\ge 0$ for at least one choice of $X=B,C$. For example, when the conservation law holds then $\delta^{(a)}_B(E)+\delta^{(a)}_C(E)=\delta^{(a)}_B(D)+\delta^{(a)}_C(D)$, implies that both $\delta^{(a)}_B(D)$ and $\delta^{(a)}_C(D)$ cannot be negative.

\textit{Examples.---}In order to investigate  how the theorems work,  we provide some illustrative examples. (i) First, let us consider  the tripartite mixed state  $\rho^1_{ABC}=\ket{\psi^1_A}\bra{\psi^1_A}\otimes\varrho^1_{BC}$. Since $A$ is factorized from the rest of the system, so that  $E(\rho^1_{AB})=D^{(C)}(\rho^1_{AC})=0$.  On the other hand, any purification of this state  leads to $\rho^1_{ABC}\rightarrow\ket{\psi^1_{ABCE}}\bra{\psi^1_{ABCE}}=\ket{\psi^1_A}\bra{\psi^1_A}\otimes \ket{\psi^1_{BCE}}\bra{\psi^1_{BCE}}$, where $\ket{\psi^1_{BCE}}$ is a purification of $\varrho^1_{BC}$. It follows that, $A$ is factorized also from $\widetilde{B}=BE$ and $\widetilde{C}=CE$, so that  $E(\rho^1_{A\widetilde{B}})=D^{(\widetilde{C})}(\rho^1_{A\widetilde{C}})=0$.  Therefore, according to the firs line of Eq. \eqref{DeltaACorrelations},  $\boldsymbol{T}^{(a)}(\rho^1_{ABC})=0$. Moreover, for this mixed state the K-W relation is satisfied with equality, reduces  in this case to $0=0+0$.

(ii) Now, as the second example, consider the  tripartite mixed state $\rho^2_{ABC}=\ket{\psi^2_{AB}}\bra{\psi^2_{AB}}\otimes\varrho^2_{C}$. Purification of this state leads to $\rho^2_{ABC}\rightarrow \ket{\psi^2_{AB}}\bra{\psi^2_{AB}}\otimes \ket{\psi^2_{CE}}\bra{\psi^2_{CE}}$, with $\ket{\psi^2_{CE}}$ as a purification of $\varrho^2_{C}$. Clearly, we  find
$E(\rho^2_{A\widetilde{B}})=E(\rho^2_{AB})=S(\rho^2_{A})$  and  $D^{(\widetilde{C})}(\rho^2_{A\widetilde{C}})=D^{(C)}(\rho^2_{AC})=0$, where $\rho^2_{A}=\Tr[\ket{\psi^2_{AB}}\bra{\psi^2_{AB}}]$. In this case, we arrive again at  $\boldsymbol{T}^{(a)}(\rho^1_{ABC})=0$, and an equality for the  K-W relation \eqref{KWinequality}.

(iii) We present the final example as a convex combination of the states given above, i.e.,
\begin{eqnarray}\label{Exampleiii}
\rho_{ABC}&=&p_1\rho^1_{ABC}+p_2\rho^2_{ABC} \\ \nonumber
&=&p_1\ket{\psi^1_A}\bra{\psi^1_A}\otimes\varrho^1_{BC}+p_2\ket{\psi^2_{AB}}\bra{\psi^2_{AB}}\otimes\varrho^2_{C},
\end{eqnarray}
where $p_1+p_2=1$. Moreover,  we set $\dim{\mathcal{H}_A}=2$ and $\dim{\mathcal{H}_B}=\dim{\mathcal{H}_C}=4$, and define
\begin{eqnarray}
\ket{\psi^1_{A}}&=&\alpha_1\ket{0_A}+\beta_1\ket{1_A}, \\
\ket{\psi^2_{AB}}&=&\alpha_2\ket{0_A0_B}+\beta_2\ket{1_A\phi_B},
\end{eqnarray}
where $\ket{\phi_B}=a\ket{1_B}+b\ket{2_B}$ and
\begin{eqnarray}
\varrho^1_{BC}&=&\lambda_1\ket{2_B2_C}\bra{2_B2_C}+(1-\lambda_1)\ket{3_B3_C}\bra{3_B3_C}, \\
\varrho^2_{C}&=&\lambda_2\ket{0_C}\bra{0_C}+(1-\lambda_2)\ket{1_C}\bra{1_C}.
\end{eqnarray}
Without loss of generality we assume that all parameters are real.
In this case, we find
\begin{eqnarray}\label{ExampleiiiAB}
\rho_{AB}&=&p_1\ket{\psi^1_A}\bra{\psi^1_A}\otimes\varrho^1_{B}+p_2\ket{\psi^2_{AB}}\bra{\psi^2_{AB}}, \\ \label{ExampleiiiAC}
\rho_{AC}&=&p_1\ket{\psi^1_A}\bra{\psi^1_A}\otimes\varrho^1_{C}+p_2\sigma^2_A\otimes\varrho^2_{C},
\end{eqnarray}
where $\sigma^2_A=\Tr_B[\ket{\psi^2_{AB}}\bra{\psi^2_{AB}}]=\alpha_2^2\ket{0_A}\bra{0_A}+\beta_2^2\ket{1_A}\bra{1_A}$.
Using Eq. \eqref{DeltaARight-Left}, one can easily calculate $\boldsymbol{T}^{(a)}(\rho_{ABC})$, where after some simplification takes the form
\begin{eqnarray}\label{ExampleiiiT}\nonumber
\boldsymbol{T}^{(a)}(\rho_{ABC})&=&\sum_{j=1}^4(-)^j\mu_j\log{\mu_j} \\
&-&p_2\beta_2^2\log{p_2\beta_2^2}+p_2\log{p_2}.
\end{eqnarray}
\begin{figure}[t!]
\centering
\includegraphics[width=10cm]{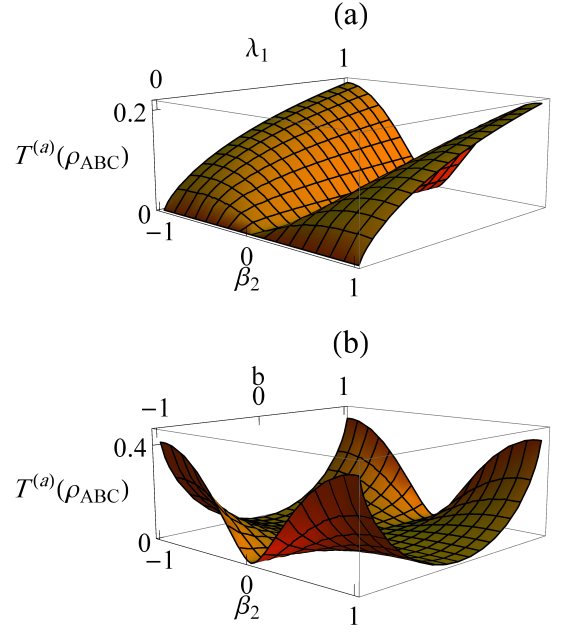}
\caption{(Color online) (a) $\boldsymbol{T}^{(a)}(\rho_{ABC})$ in terms of $\beta_2$ and $\lambda_1$, for $b=1/\sqrt{2}$. (b) $\boldsymbol{T}^{(a)}(\rho_{ABC})$ in terms of $\beta_2$ and $b$, for $\lambda_1=1/2$. In both plots, we assumed $\alpha_1=1/\sqrt{2}$ and  $p_1=\lambda_2=1/2$.}
\label{Figure1}
\end{figure}
Here
\begin{eqnarray}
\mu_{1,3}&=&\frac{1}{2}\left[p_1\lambda_1+p_2\pm\sqrt{(p_1\lambda_1-p_2)^2+4p_1p_2\beta_1^2\gamma^2}\right], \\
\mu_{2,4}&=&\frac{1}{2}\left[p_1\lambda_1+p_2\beta_2^2\pm\sqrt{(p_1\lambda_1-p_2\beta_2^2)^2+4p_1p_2\gamma^2}\right],
\end{eqnarray}
where $\gamma=\sqrt{\lambda_1} b\beta_2$. Simple investigation of Eq. \eqref{ExampleiiiT} shows that $\boldsymbol{T}^{(a)}(\rho_{ABC})$ vanishes  if and only if $\gamma=0$, i.e., one of the parameters  $\lambda_1$, $b$, or $\beta_2$ vanishes. In Fig. \ref{Figure1} we have plotted $\boldsymbol{T}^{(a)}(\rho_{ABC})$ in terms of the pairs $\{\beta_2,\; \lambda_1\}$ and $\{\beta_2, \; b\}$. The figure shows clearly that  $\boldsymbol{T}^{(a)}(\rho_{ABC})$ approaches zero whenever one of the parameters    $\lambda_1$, $b$,  or $\beta_2$ approaches zero. Now, let us turn our attention to Theorem \ref{TheoremSSAcondition} and gain a better understanding of  this Theorem.   To this aim, first note that each term of Eq. \eqref{Exampleiii} has vanishing $\boldsymbol{T}^{(a)}(\rho_{ABC})$, i.e., $\boldsymbol{T}^{(a)}(\rho^1_{ABC})=\boldsymbol{T}^{(a)}(\rho^2_{ABC})=0$. Moreover, using
\begin{eqnarray}
\rho^1_{B}&=& \lambda_1\ket{2_B}\bra{2_B}+(1-\lambda_1)\ket{3_B}\bra{3_B}, \\
\rho^2_{B}&=&\alpha_2^2\ket{0_B}\bra{0_B}+\beta_2^2\ket{\phi_B}\bra{\phi_B}, \\
\rho^1_{C}&=&\lambda_1\ket{2_C}\bra{2_C}+(1-\lambda_1)\ket{3_C}\bra{3_C}, \\
\rho^2_C&=&\varrho^2_C,
\end{eqnarray}
one can easily see that $\rho^1_B\rho^2_B=\lambda_1b\beta_2^2 \ket{2_B}\bra{\phi_B}$ and $\rho^1_C\rho^2_C=0$. This implies that the required conditions of Theorem \ref{TheoremSSAcondition} are satisfied if  and only if $\lambda_1 b \beta_2=0$, which is in  complete agreement with the result obtained  from Eq. \eqref{ExampleiiiT}.

We continue with this example  and apply Theorem \ref{TheoremDeltaCorrelations} and Corollary \ref{Corollary-Results} to evaluate QD and EOF of $2\times 4$ mixed bipartite  states $\rho_{AB}$ and $\rho_{AC}$, reduced from \emph{mixed} tripartite state $\rho_{ABC}$ of Eq.  \eqref{Exampleiii}.   Equation  \eqref{ExampleiiiAC} shows that  $\rho_{AC}$ is separable  and, since  $\rho^1_{C}\perp\rho^2_C$,  we get \cite{ZurekPRL2001}
\begin{equation}\label{ExampleiiiCorrelationAC}
E(\rho_{AC})=D^{(C)}(\rho_{AC})=0,
\end{equation}
which holds for  arbitrary values of $\boldsymbol{T}^{(a)}(\rho_{ABC})$. However, for $\boldsymbol{T}^{(a)}(\rho_{ABC})=0$, i.e., $\lambda_1 b \beta_2=0$, we can use the Corollary \ref{Corollary-Results} and write
\begin{eqnarray}\label{ExampleiiiCorrelationAB}
E(\rho_{AB})=D^{(B)}(\rho_{AB})=S_{\rho_{AC}}(A\mid C).
\end{eqnarray}
Here,  the first equality is obtained from Eq. \eqref{ExampleiiiCorrelationAC} and  the conservation law \ref{Corollary-Results}-(ii),  and the second equality comes from the K-W relation \ref{Corollary-Results}-(i). Furthermore, denoting an orthonormal basis for the auxiliary subsystem $E$ with $\{\ket{\lambda_E}\}_{\lambda=0}^{3}$,  one can provide the following purification for $\rho_{ABC}$
\begin{eqnarray} \nonumber
\ket{\Psi_{ABCE}}&=&\sqrt{p_1}\ket{\psi^1_A}\left[\sqrt{\lambda_1}\ket{2_B2_C2_E}+\sqrt{1-\lambda_1}\ket{3_B3_C3_E}\right] \\
&+&\sqrt{p_2}\ket{\psi^2_{AB}}\left[\sqrt{\lambda_2}\ket{0_C0_E}+\sqrt{1-\lambda_2}\ket{1_C1_E}\right].
\end{eqnarray}
Using this we find
\begin{eqnarray}\label{ExampleiiiRhoABE}
\rho_{A\widetilde{B}}&=&p_1\ket{\psi^1_A}\bra{\psi^1_A}\otimes\varrho^1_{BE}+p_2\ket{\psi^2_{AB}}\bra{\psi^2_{AB}}\otimes\varrho^2_{E}, \\ \label{ExampleiiiRhoACE}
\rho_{A\widetilde{C}}&=&p_1\ket{\psi^1_A}\bra{\psi^1_A}\otimes\varrho^1_{CE}+p_2\sigma^2_{A}\otimes\varrho^2_{CE} \\ \nonumber
&+&\sqrt{p_1p_2\lambda_1}b\left(\ket{1_A}\bra{\psi^1_A}\otimes \ket{\psi_{CE}}\bra{2_C2_E}\right. \\ \nonumber
&& \qquad\qquad+\left.\ket{\psi^1_A}\bra{1_A}\otimes \ket{2_C2_E}\bra{\psi_{CE}}\right),
\end{eqnarray}
where we have defined
\begin{eqnarray}
\varrho^1_{BE}&=&\lambda_1\ket{2_B2_E}\bra{2_B2_E}+(1-\lambda_1)\ket{3_B3_E}\bra{3_B3_E}, \\
\varrho^1_{CE}&=&\lambda_1\ket{2_C2_E}\bra{2_C2_E}+(1-\lambda_1)\ket{3_C3_E}\bra{3_C3_E}, \\
\varrho^2_{CE}&=&\lambda_2\ket{0_C0_E}\bra{0_C0_E}+(1-\lambda_2)\ket{1_C1_E}\bra{1_C1_E}, \\
\varrho^2_{E}&=&\lambda_2\ket{0_E}\bra{0_E}+(1-\lambda_2)\ket{1_E}\bra{1_E}, \\
\ket{\psi_{CE}}&=&\sqrt{\lambda_2}\ket{0_C0_E}+\sqrt{1-\lambda_2}\ket{1_C1_E}.
\end{eqnarray}
The bipartite states \eqref{ExampleiiiRhoABE} and \eqref{ExampleiiiRhoACE} are obtained for  arbitrary values of $\boldsymbol{T}^{(a)}(\rho_{ABC})$. However, if we set one of the parameters  $\lambda_1$, $b$,  or $\beta_2$ equal to zero, we get $\boldsymbol{T}^{(a)}(\rho_{ABC})=0$ and one can use the benefits of invariance  of quantum  correlations of $\rho_{AB}$ and $\rho_{AC}$  under the transformation $B\longrightarrow BE$ and $C\longrightarrow CE$. In this case,  invoking   Eqs. \eqref{ExampleiiiCorrelationAB} and \eqref{ExampleiiiCorrelationAC}, one can  write
\begin{eqnarray}
E(\rho_{A\widetilde{B}})&=&D^{(\widetilde{B})}(\rho_{A\widetilde{B}})=S_{\rho_{AC}}(A \mid C), \\  E(\rho_{A\widetilde{C}})&=&D^{(\widetilde{C})}(\rho_{A\widetilde{C}})=0,
\end{eqnarray}
which are valid as far as  $b\lambda_1\beta_2=0$.

\textit{Conclusion.---}We have defined $\boldsymbol{T}^{(a)}(\rho_{ABC})$ as the extent to which the  tripartite state $\rho_{ABC}$  fails to saturate  SSA inequality.    An important feature of  our approach  is the possibility of writing $\boldsymbol{T}^{(a)}(\rho_{ABC})$ as the amount by which
the bipartite quantum correlations of $\rho_{AB}$ and $\rho_{AC}$  change under the transformation $B\rightarrow \widetilde{B}=BE$  and $C\rightarrow \widetilde{C}=CE$, with  $E$ as an auxiliary subsystem purifying $\rho_{ABC}$.
This feature of $\boldsymbol{T}^{(a)}(\rho_{ABC})$ seems  remarkable since it  provides a simple method to  find the structure of states for which the SSA inequality is saturated by its lower and upper bounds.  The  concavity property of $\boldsymbol{T}^{(a)}(\rho_{ABC})$  with respect to its input reveals that  a state with vanishing $\boldsymbol{T}^{(a)}(\rho_{ABC})$ can  be realized \emph{only} by  those ensembles $\{p_k,\rho_{ABC}^k\}_k$  for which   $\boldsymbol{T}^{(a)}(\rho_{ABC}^k)=0$. We have characterized such ensembles by means of invariance of quantum correlations of  $\rho_{AB}$ and $\rho_{AC}$ under  the transformation $B\rightarrow \widetilde{B}$ and $C\rightarrow \widetilde{C}$.
It turns out   that   $\boldsymbol{T}^{(a)}(\rho_{ABC}^k)=0$ if and only if the subsystem $A$ lives in a \emph{pure} state, so that it  is not affected under purification of $\rho^k_{ABC}$.   On the other hand, the upper bound $\boldsymbol{T}^{(a)}(\rho_{ABC})=2\log{d_A}$ is saturated if and only if the subsystem  $A$ participates in the purification of $\rho_{ABC}$ as possible as it can, i.e.,  $\rho_{A}$ is  maximally mixed and factorized from the rest of the system.  Intuitively, the contribution of the subsystem $A$ in the purification of $\rho_{ABC}$ plays a central role in a sense that the amount by which the quantum correlations of $\rho_{AB}$ and $\rho_{AC}$ changes under the above transformation depends on the extent to which   the subsystem $A$ shares its degrees of freedom in the purification. It happens that,    the more contribution the subsystem $A$ has in the purification of $\rho_{ABC}$, the more quantum correlations will be  shared between $A$ and the subsystems  $\widetilde{B}$ and $\widetilde{C}$, leading  to a greater value for $\boldsymbol{T}^{(a)}(\rho_{ABC})$.

Moreover, the approach presented in this paper explores  that the class of states for which the SSA inequality is saturated  coincides with  those that the K-W inequality is saturated.  This generalizes, to the best of our knowledge, the previous results for pure states to those with vanishing $\boldsymbol{T}^{(a)}(\rho_{ABC})$. Interestingly, the condition $\boldsymbol{T}^{(a)}(\rho_{ABC})=0$ exhausts such states.       In addition, we found that if $S(\rho_{B})+S(\rho_{C}) = S(\rho_{AB})+S(\rho_{AC})$ then   $E(\rho_{AB})+E(\rho_{AC})=D^{(B)}(\rho_{AB})+D^{(C)}(\rho_{AC})$, which is  a possible extension  of the so-called quantum conservation law previously obtained  in  \cite{FanchiniPRA2011} for pure states.   Due to the widespread use of the SSA inequality in quantum information theory, it is hoped that a quantum correlation description of SSA inequality should shed some light on the several inequalities obtained from it. In particular, our results may have applications in  monogamy inequalities, squashed entanglement, Holevo bounds, coherent information, and study of open quantum systems.

The authors  would like to  thank  Fereshte Shahbeigi and Karol \.{Z}yczkowski for  helpful  discussion and comments. This work was supported by Ferdowsi University of Mashhad under grant 3/28328 (1392/07/15).

\end{document}